\documentclass[conference,10pt]{IEEEtran}
\usepackage[dvips]{graphicx}
\usepackage[cmex10]{amsmath}
\usepackage{multirow}
\usepackage{epsfig}
\usepackage{mathrsfs}
\usepackage{amssymb}
\usepackage{comment}
\usepackage{bm}
\usepackage{array}
\usepackage{amsthm}
\usepackage{blkarray}
\usepackage{enumerate}
\usepackage{url}
\usepackage{chemarrow}
\usepackage{fancyhdr}
\usepackage[lined,ruled,linesnumbered]{algorithm2e}
\usepackage{epsf,psfrag}
\usepackage{epsfig}

\renewcommand{\thispagestyle}[1]{} % do nothing

\theoremstyle{definition}
\newtheorem{theorem}{Theorem}[section]
\newtheorem{lemma}[theorem]{Lemma}

\newtheorem{definition}{Definition}

\def\tSTPC{t_{\mbox{\tiny STPC}}}
\def\tRWPC{t_{\mbox{\tiny RWPC}}}
\def\tSTLI{t_{\mbox{\tiny STLI}}}
\def\tMILI{t_{\mbox{\tiny MILI}}}
\def\hSTPC{h_{\mbox{\tiny STPC}}}
\def\hRWPC{h_{\mbox{\tiny RWPC}}}

\ifCLASSINFOpdf

\else

\fi

\usepackage[tight,footnotesize]{subfigure}

\begin{document}

\pagestyle{fancy}
\IEEEoverridecommandlockouts

\lhead{\textit{Technical Report, Dept. of EEE, Imperial College, London, UK, Nov., 2012.}}
\rhead{} %\thepage can be added to the brackets
\title{Efficient Identification of Additive Link Metrics: Theorem Proof and Evaluations}
%\author{\IEEEauthorblockN{A\IEEEauthorrefmark{2}, B\IEEEauthorrefmark{3}, C\IEEEauthorrefmark{1} and D\IEEEauthorrefmark{4}}
%\IEEEauthorblockA{\IEEEauthorrefmark{2}E\\
%\IEEEauthorrefmark{3}F\\
%\IEEEauthorrefmark{1}G\\
%\IEEEauthorrefmark{4}H
\author{\IEEEauthorblockN{Liang Ma\IEEEauthorrefmark{2}, Ting He\IEEEauthorrefmark{3}, Kin K. Leung\IEEEauthorrefmark{2}, Don Towsley\IEEEauthorrefmark{1}, and Ananthram Swami\IEEEauthorrefmark{4}}
\IEEEauthorblockA{\IEEEauthorrefmark{2}Imperial College, London, UK. Email: \{l.ma10, kin.leung\}@imperial.ac.uk\\
\IEEEauthorrefmark{3}IBM T. J. Watson Research Center, Yorktown, NY, USA. Email: the@us.ibm.com\\
\IEEEauthorrefmark{4}Army Research Laboratory, Adelphi, MD, USA. Email: ananthram.swami.civ@mail.mil\\
\IEEEauthorrefmark{1}University of Massachusetts, Amherst, MA, USA. Email: towsley@cs.umass.edu
}
%
%\thanks{Research was partially sponsored by the U.S. Army Research Laboratory and the U.K. Ministry of Defence and was accomplished under Agreement Number
%W911NF-06-3-0001.
%%The views and conclusions contained in this document
%%are those of the authors and should not be interpreted as representing the
%%official policies, either expressed or implied, of the U.S. Army Research
%%Laboratory, the U.S. Government, the U.K. Ministry of Defence or the U.K.
%%Government. The U.S. and U.K. Governments are authorized to reproduce and
%%distribute reprints for Government purposes notwithstanding any copyright
%%notation hereon.
%}
}

\maketitle

\IEEEpeerreviewmaketitle

\section{Introduction}
As described in \cite{LiangPathConstruction12}, algorithm STPC (Spanning Tree-based Path Construction) builds three independent spanning trees \cite{GraphTheory2005} ($\mathcal{T}_1$, $\mathcal{T}_2$ and $\mathcal{T}_3$) wrt node $r$ in the 3-vertex-connected graph $\mathcal{G}^*_{ex}$ ($\mathcal{G}^*_{ex}$ is obtained by adding virtual nodes and virtual links to the original graph $\mathcal{G}$) for constructing a sufficient number of simple paths such that all links in $\mathcal{G}$ can be identified. Based on these three independent spanning trees, STPC constructs 3 internally vertex disjoint paths wrt each node in $\mathcal{G}$, thus resulting in $3|V|$ paths totally, where $V$ is the set of nodes in $\mathcal{G}$. However, there exists redundant information in these $3|V|$ paths since some of them are exactly the same. The main objective of this report is to prove that the number of distinct paths constructed by STPC in fact equals $|L|$, where $L$ is the set of links in $\mathcal{G}$. By linear algebra, we know that the minimum number of the required path measurements for full link identifications is $|L|$ and all these paths must be linearly independent. Therefore, STPC is characterized by high efficiency in that all constructed distinct paths are linearly independent, i.e., all link metrics in $\mathcal{G}$ can be identified without conducting any unnecessary path measurements. In the end, we show some additional simulation results to verify the superior efficiency of STPC and STLI. The terms and notations used in this report are defined in Table \ref{t notion}.

\begin{table}[b]
%\vspace{-.5em}
\renewcommand{\arraystretch}{1.3}
\caption{Main Notations} \label{t notion}
\vspace{-.5em}
\centering
\begin{tabular}{r|m{6.5cm}}
  \hline
  \textbf{Symbol} & \textbf{Meaning} \\
  \hline
  $V(\mathcal{G})$, $L(\mathcal{G})$ & set of nodes/links in graph $\mathcal{G}$\\
  \hline
  $m$, $n$ & number of nodes/links in $\mathcal{G}$ \\
  \hline
%  $L(v)$ & set of links incident to node $v$\\
%  \hline
  %$E(\mathcal{G})$ & set of exterior links (see Definition\ref{def:interior graph/links}) in graph $\mathcal{G}$\\
%  \hline
%  $\mathcal{G}-l$ & delete a link: $\mathcal{G}-l=(V(\mathcal{G}),L(\mathcal{G})\setminus \{l\})$, where $l\in L(\mathcal{G}$) and ``$\setminus$'' is setminus\\
%  \hline
  $\mathcal{G}+l$ & add a link: $\mathcal{G}+l=(V(\mathcal{G}),L(\mathcal{G})\cup \{l\})$, where the end-points of link $l$ are in $V(\mathcal{G})$\\
 \hline
  $\mathcal{G} \setminus \mathcal{G}^{'}$ & From $\mathcal{G}$, delete all nodes in common with $\mathcal{G}'$ and their incident links \\
  \hline
   $\mathcal{G} \cup \mathcal{G}^{'} $ & graph union: $ \mathcal{G} \cup \mathcal{G}^{'}=(V(\mathcal{G}) \cup V(\mathcal{G}'), L(\mathcal{G}) \cup L(\mathcal{G}'))$ \looseness=-1\\
  \hline
  $\mathcal{P}$ & simple path, defined as a graph with $V(\mathcal{P})=\{v_0,\ldots,v_k\}$ and $L(\mathcal{P})=\{v_0v_1, v_1v_2,\ldots,v_{k-1}v_k\}$, where $v_0,\ldots,v_k$ are distinct nodes\\
  \hline
  ${\mu}_i$ & ${\mu}_i\in V(\mathcal{G})$ is the $i$-th monitor in $\mathcal{G}$ \\
  \hline
  $w_l$, $c_{\mathcal{P}}$ &  metric of link $l$, sum metric of path $\mathcal{P}$ \looseness=-1\\
  \hline

\end{tabular}
\vspace{-0mm}
\end{table}

\section{Independent Spanning Trees}
%Let $\mathcal{G}_m$ denote $\mathcal{T}_1\cup \mathcal{T}_2\cup \mathcal{T}_3$ obtained by STPC wrt $\mathcal{G}^*_{ex}$. Thus, each link in $\mathcal{G}_m$ is a tree link. According to Theorem 2 in \cite{Cheriyan88}, $\mathcal{G}_m$ is 3-vertex-connected.
With respect to $\mathcal{G}^*_{ex}$, let $\mathcal{P}_{STPC}$ denote the path set obtained by running STPC algorithm. As discussed before, the goal is to prove that the number of distinct paths in $\mathcal{P}_{STPC}$ is the number of links in $\mathcal{G}$, i.e., $|\mathcal{P}_{STPC}|=|L|$. For this purpose, we first briefly discuss how to construct the 3 independent spanning trees, some unique features of which will be used in later proofs.

\subsection{Three Independent Spanning Trees}
In this section, we consider how to construct three independent spanning trees wrt node $r$ in $\mathcal{G}^*_{ex}$.

\begin{definition} (\cite{Cheriyan88})
\emph{Nonseparating ear decomposition:} Nonseparating ear decomposition of $\mathcal{G}$ is a decomposition $V(\mathcal{G})=V(\mathcal{E}_1\cup \mathcal{E}_2 \cup \cdots \cup \mathcal{E}_{n_e})$ (each $\mathcal{E}_i$ is called an \emph{ear}, $n_e$ ears in total) such that
\begin{enumerate}
  \item $\mathcal{E}_1$ is an induced cycle,
  \item $\mathcal{E}_i$ ($2\leq i \leq n_e$) is an induced simple path with only its end-points in common with $\mathcal{E}_1 \cup \cdots \cup \mathcal{E}_{i-1}$,
  \item Let $\overline{\mathcal{G}_i}:=\mathcal{G}\setminus (\mathcal{E}_1\cup \mathcal{E}_2 \cup \cdots \cup \mathcal{E}_{i})$. For each $\mathcal{E}_i$, $\overline{\mathcal{G}_i}$ is connected and each internal node of $\mathcal{E}_i$ has a neighbor in $\overline{\mathcal{G}_i}$, and
  \item $|\mathcal{E}_{n_e}|=3$ and the internal node of $\mathcal{E}_{n_e}$ does not appear in other ears.
\end{enumerate}
\end{definition}

\begin{definition} (\cite{Even76})
\emph{s-t numbering:} s-t numbering of $\mathcal{G}$ is a one-one function $f:\ V\rightarrow [1\cdots n]$ ($n=|V|$) such that
\begin{enumerate}
  \item $f(s)=1$, $f(t)=n$, and
  \item For each vertex $v$ in $\mathcal{G}\setminus \{s,t\}$, it has a neighbor $u$ with $f(u)<f(v)$ and a neighbor $w$ with $f(w)>f(v)$.
\end{enumerate}
\end{definition}

Suppose $\mathcal{E}_j$ is the first ear with node $v$ on it, then $j$ is called the \emph{ear level} of $v$, denoted by $g(v)=j$. It is proved in \cite{Cheriyan88} that a 3-vertex-connected graph $\mathcal{G}^*_{ex}$ has a nonseparation ear decomposition, by which each node in $\mathcal{G}^*_{ex}$ can be assigned an ear level and an s-t number. With the knowledge of the ear level and the s-t number of each node, three independent spanning trees can be constructed accordingly. Without loss of generality, in the sequel, we assume ${\mu}_1$ is a non-cutvertex monitor\footnote{A non-cutvetex monitor can be found since it it impossible that all monitors are cutvertices in an identifiable network according to the minimum monitor placement algorithm MMP \cite{MaInfocomOnline}.} in $\mathcal{G}$ and virtual node $r$ connects to ${\mu}_2$, i.e., in the definition \cite{LiangPathConstruction12} of $\mathcal{G}^*_{ex}$, ${\mu}_i$ equals ${\mu}_2$. Due to the special structure of $\mathcal{G}^*_{ex}$, the complicated algorithm  \cite{Cheriyan88} of finding $\mathcal{E}_1$ can be avoided, i.e., simply choose $r{\mu'}_1{\mu}_1{\mu}'_2r$ as $\mathcal{E}_1$. Then the rest ears can be found by the ear decomposition algorithm described in the proof of Theorem 1 \cite{Cheriyan88}. To construct 3 independent spanning trees, one rule for selecting the last ear $\mathcal{E}_{n_e}$ is that its internal node must be ${\mu}_2$. One preferable property of nonseparation ear decomposition is that the s-t numbers can be naturally computed in the process of selecting each ear, following Algorithm \ref{alg_stNum}.%\looseness=-1

\begin{algorithm}[tb]
\small
\SetKwInOut{Input}{input}\SetKwInOut{Output}{output}
\Input{Ear decomposition of graph $\mathcal{G}$}
\Output{s-t number of each node in $\mathcal{G}$}
%\BlankLine
%$\textbf{w}=\mathbf{0}$, $\textbf{W}=\mathbf{0}$\;
\ForEach {ear $\mathcal{E}_i$}
    {
    \eIf {i=1}
        {
        $f(r)=1$, $f({\mu}'_2)=2$, $f({\mu}_1)=3$, $f({\mu}'_1)=4$\;
        }
        {
        Suppose the end-points of $\mathcal{E}_i$ are $v_1$ and $v_2$ with $f(v_1)<f(v_2)$ ($\eta = f(v_2)$) and $\psi = |\mathcal{E}_i|-2$, then\\
  (i) for each existing node $w$ with $f(w)\geq\eta$, $f(w)\leftarrow f(w)+\psi$, and\\
  (ii) following the direction from $v_1$ to $v_2$ on $\mathcal{E}_i$, all internal nodes of $\mathcal{E}_i$ are sequentially numbered from $\eta$ to $\eta+\psi-1$.
        }
    }
\caption{s-t numbering of $\mathcal{G}$}\label{alg_stNum}
\vspace{-.25em}
\end{algorithm}
\normalsize

By running Algorithm \ref{alg_stNum}, the final s-t numbers for $r$ and ${\mu}'_1$ are $1$ and $n$ (the number of nodes in $\mathcal{G}^*_{ex}$), respectively. With the computed ear level and s-t numbers of a given 3-vertex-connected graph $\mathcal{G}^*_{ex}$, it is easy to show that each node $v$ in $\mathcal{G}^*_{ex}\setminus \{r,{\mu}'_1,{\mu}_2\}$ has three neighbors $w_1$, $w_2$ and $w_3$ such that
\begin{enumerate}
  \item ear level $g(w_1)>g(v)$;
  \item $f(w_2)<f(v)$ and $g(w_2)\leq g(v)$, and
  \item $f(w_3)>f(v)$ and $g(w_3)\leq g(v)$.
\end{enumerate}

The above three properties are also the three rules to construct the corresponding independent spanning trees:
\begin{itemize}
  \item $\mathcal{T}_1$: First, add link $r\mu_2$ to $\mathcal{T}_1$. Next, for each node $v$ in $\mathcal{G}_m\setminus \{r, {\mu}_2\}$, $\mathcal{T}_1$ involves link $vw$, where $w$ is a neighbor node of $v$ with $g(w)>g(v)$;
  \item $\mathcal{T}_2$: First, add link $r{\mu'}_2$ to $\mathcal{T}_2$. Next, for each node $v$ in $\mathcal{G}_m\setminus \{r, {\mu}'_2\}$, $\mathcal{T}_2$ involves link $vw$, where $w$ is a neighbor node of $v$ with $f(w)<f(v)$ and $g(w)\leq g(v)$;
  \item $\mathcal{T}_3$: First, add link $r{\mu'}_1$ to $\mathcal{T}_3$. Next, for each node $v$ in $\mathcal{G}_m\setminus \{r, {\mu}'_1\}$, $\mathcal{T}_3$ involves link $vw$, where $w$ is a neighbor node of $v$ with $f(w)>f(v)$ and $g(w)\leq g(v)$.
\end{itemize}

Accordingly, three independent spanning trees $\mathcal{T}_1$, $\mathcal{T}_2$ and $\mathcal{T}_3$ are constructed in $\mathcal{G}^*_{ex}$.

\section{Number of Distinct Paths Constructed by STPC}
With the constructed three independent spanning trees, let $\mathcal{G}_m:=\mathcal{T}_1\cup \mathcal{T}_2\cup \mathcal{T}_3$. Then each node in $\mathcal{G}_{m}\setminus r$ (virtual nodes ${\mu}'_1$ and ${\mu}'_2$ are also included) has three internally vertex disjoint paths (denoted by $\mathcal{X}^{(1)}$, $\mathcal{X}^{(2)}$ and $\mathcal{X}^{(3)}$) to $r$, each along the corresponding spanning tree. Combining any two of these three paths, we get three cycles, i.e., $\mathcal{C}^{(1)}=\mathcal{X}^{(1)}\cup \mathcal{X}^{(2)}$, $\mathcal{C}^{(2)}=\mathcal{X}^{(2)}\cup \mathcal{X}^{(3)}$ and $\mathcal{C}^{(3)}=\mathcal{X}^{(1)}\cup \mathcal{X}^{(3)}$. Let $C$ be the set of all these cycles. To count the number of distinct paths constructed by STPC, we first count the number of distinct cycles in $C$. Next, removing all virtual links and the resulting isolated nodes in each cycle $\mathcal{C}_i$ of $C$, one monitor-to-monitor simple path $\mathcal{Y}_i$ (if $||\mathcal{Y}_i||\geq 1$) is formed, thus generating a path set ${Y}$. We then study the number of linearly independent paths in ${Y}$ and yield path set ${Y}'$ by removing the linearly dependent paths in ${Y}$. Finally, we show that for each path $\mathcal{Y}'_i$ with more than 2 monitors in ${Y}'$, it can be shortened to a path, denoted by $\mathcal{Z}_i$, with only 2 monitors while retaining the property of full tree link identification in $\mathcal{G}^*_{ex}$. We can prove path set $\{\mathcal{Z}_i\}$ has $||\mathcal{G}_m||-|\mathbb{V}|$ paths ($\mathbb{V}$ is the set of virtual links in $\mathcal{G}_m$), all of which are linearly independent and can cover all paths constructed by STPC for tree link identification in $\mathcal{G}$. Therefore, for tree link identification, the number of distinct paths constructed by STPC is exactly the number of tree links. In the end, we show that the auxiliary algorithm of STPC constructs one path for each non-tree link in $\mathcal{G}$, thus completing the proof of Theorem IV.2 in \cite{LiangPathConstruction12} as a link in $\mathcal{G}$ is either a tree or non-tree link.

\subsection{Number of Distinct Cycles}
\label{sec_numDisCycl}
Let $\mathcal{X}^{(1)}_v$, $\mathcal{X}^{(2)}_v$ and $\mathcal{X}^{(3)}_v$ denote three internally vertex disjoint paths from $v$ ($v\in(\mathcal{G}_m\setminus r)$) to $r$ along the corresponding independent spanning trees. Combining any two of these paths, node $v$ is associated with three cycles, i.e., $\mathcal{C}^{(1)}_v=\mathcal{X}^{(1)}_v\cup \mathcal{X}^{(2)}_v$, $\mathcal{C}^{(2)}_v=\mathcal{X}^{(2)}_v\cup \mathcal{X}^{(3)}_v$ and $\mathcal{C}^{(3)}_v=\mathcal{X}^{(1)}_v\cup \mathcal{X}^{(3)}_v$. Therefore, there are $3(|\mathcal{G}_m|-1)$ cycles, forming cycle set $C$. In this section, we investigate on counting the distinct cycles in $C$. For illustrative purpose, we color the links on $\mathcal{T}_1$ as blue, $\mathcal{T}_2$ as green, and $\mathcal{T}_3$ as red in the sequel. Since the independent trees are constructed wrt $r$, each link in $\mathcal{T}_1$, $\mathcal{T}_2$ and $\mathcal{T}_3$ has a natural direction toward $r$ when selecting $\mathcal{X}^{(i)}_v$ ($i=1,2,3$). Therefore, for the simplicity of the following proofs, each link in $\mathcal{T}_1$, $\mathcal{T}_2$ and $\mathcal{T}_3$ is assigned a direction toward $r$. Note this direction is for and only for the theorem proof and only exists on $\mathcal{T}_1$, $\mathcal{T}_2$ and $\mathcal{T}_3$, meaning the assumption that the original graph is undirected still holds. With this notation, links $l_1$ on $\mathcal{T}_i$ and $l_2$ on $\mathcal{T}_j$ ($i\neq j$) with different colors and different directions might correspond to the same link in $\mathcal{G}_m$.

\begin{lemma}
The number of distinct cycles in $C$ is the number of links (both real and virtual) in $\mathcal{G}_m$, i.e., $|C|=||\mathcal{G}_m||$. Furthermore, all these distinct cycles are linearly independent.
\end{lemma}

\begin{proof}
In $\mathcal{G}_m$, let $b_1$ be the number of links appearing only in $\mathcal{T}_1$ (colored as blue), and $b_2$ be the number of links appearing in two trees: one is $\mathcal{T}_1$ (colored as blue), the other one is either $\mathcal{T}_2$ (colored as green) or $\mathcal{T}_3$ (colored as red). Then we have:
\begin{equation}\label{eq_red_num}
    b_1 + b_2 = |\mathcal{G}_m|-1,
\end{equation}
since $\mathcal{T}_1$ is a spanning graph of $\mathcal{G}_m$. Consider ear $\mathcal{E}_i$. Let $\delta_i$ be the number of newly added nodes to $\mathcal{E}_1\cup\cdots\cup \mathcal{E}_{i-1}$ by $\mathcal{E}_i$. Let $v_{0}\cdots v_{\delta_i+1}$ denote the nodes in $\mathcal{E}_i$, where $v_0$ and $v_{\delta_i+1}$ are end-points already existed in previous ears. Then we can derive the relationship among the number of links, the number of nodes and the number of ears in $\mathcal{G}_m$. Suppose there are $n_e$ ears in total, then $\mathcal{E}_1$ with $\delta_1$ new nodes has $\delta_1$ links since $\mathcal{E}_1$ is a cycle, whereas $\mathcal{E}_i$ ($2\leq i \leq {n_e}$) which is a path with $\delta_i$ new nodes has $\delta_i+1$ new links. In addition to these links, there exist $b_1$ links only appearing in the blue tree and not involving in any ears. Hence, we have

\begin{equation}\label{eq_linkNodeRelation}
\begin{aligned}
    ||\mathcal{G}_m|| & = b_1 +\delta_1 + \sum^{n_e}_{i=2}(\delta_i+1)\\
    &=b_1+\delta_1+\sum^{n_e}_{i=2}\delta_i+\sum^{n_e}_{i=2}1\\
    &=b_1+|\mathcal{G}_m|+{n_e}-1.
\end{aligned}
\end{equation}

To count distinct cycles in $C$, we consider merging the cycles obtained from each pair of independent spanning trees.
\begin{figure}[tb]
\centering
\includegraphics[width=1.8in]{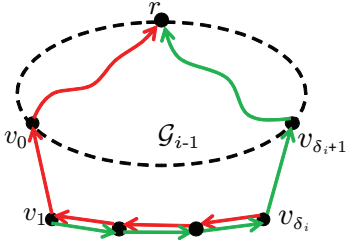}
\caption{Paths generated by merging red and green path segments.} \label{Fig:CombineRedGreen}
\end{figure}

(i) We first consider the cycles obtained by combining the red and green paths. For ear $\mathcal{E}_i$, each newly added node in $\mathcal{E}_i$ corresponds to a cycle obtained by combining the red and green paths. However, as Fig.~\ref{Fig:CombineRedGreen} shows, each link between $v_1$ and $v_{\delta_i}$ corresponds to two colors, red and green, with opposite directions. Therefore, cycles formed by combining the red and green paths wrt $v_1\cdots v_{\delta_i+1}$ are identical. Thus, ear $\mathcal{E}_i$ only contributes one distinct cycle when combining the red and green paths for any of the newly added nodes $v_1\cdots v_{\delta_i}$.
\begin{figure}[tb]
\centering
\includegraphics[width=1.8in]{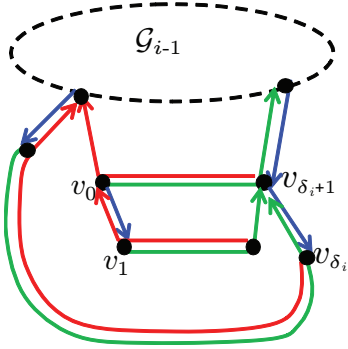}
\caption{Paths generated by merging blue and red/green path segments.} \label{Fig:CombineBlueWithOthers}
\end{figure}

(ii) Next, we consider merging the cycles generated by combining the red and the blue paths. For the $\delta_i$ new nodes in ear $\mathcal{E}_i$, each path formed by combining the red and the blue paths are distinct from each other. However, as Fig.~\ref{Fig:CombineBlueWithOthers} displays, the $red+blue$ cycle wrt $v_1$ might be identical with the $red+blue$ cycle wrt $v_0$ (when $v_0$ was first added in the previous operations). Thus, the same cycle might be counted twice. Let $\varepsilon_i$ indicate if link $v_0v_1$ is colored as blue, i.e., $\varepsilon_i=1$ if true, and $\varepsilon_i=0$ otherwise. With this notation, the number of distinct cycles formed by combining the red and the blue paths is $\delta_i-\varepsilon_i$ for ear $\mathcal{E}_i$ ($2\leq i \leq n_e $), and $\delta_1-1-\varepsilon_1$ for ear $\mathcal{E}_1$.

(iii) Finally, we consider merging the cycles formed by combining the green and the blue paths. Following the same argument in (ii), each path formed by combining the green and the blue paths are distinct from the $\delta_i$ new nodes in ear $\mathcal{E}_i$. However, the cycles wrt to $v_{\delta_i}$ and $v_{\delta_i+1}$ might be identical (see Fig.~\ref{Fig:CombineBlueWithOthers}). Let $\varepsilon'_i$ indicate if link $v_{\delta_i}v_{\delta_i+1}$ is colored blue, i.e., $\varepsilon'_i=1$ if true, and $\varepsilon'_i=0$ otherwise. Accordingly, the number of distinct cycles formed by combining the green and the blue paths is $\delta_i-\varepsilon'_i$ for ear $\mathcal{E}_i$ ($2\leq i \leq n_e$), and $\delta_1-1-\varepsilon'_1$ for ear $\mathcal{E}_1$.

Therefore, the number of distinct cycles, denoted by $Q_i$, contributed by ear $\mathcal{E}_i$ is:
\begin{equation}\label{eq_earContribution}
\begin{aligned}
    Q_1&=1+\delta_1-1-\varepsilon_1+\delta_1-1-\varepsilon'_1,\\
    Q_i&=1+\delta_i-\varepsilon_i+\delta_i-\varepsilon'_i\ \  (2\leq i \leq n_e).
\end{aligned}
\end{equation}

Thus, the number of distinct cycles in $C$ is
\begin{equation}\label{eq_totalCycles}
\begin{aligned}
    |C|&=Q_1+\sum^{n_e}_{i=2}Q_i\\
     &=2\delta_1-\varepsilon_1-\varepsilon'_1-1+2\sum^{n_e}_{i=2}\delta_i+\sum^{n_e}_{i=2}1-\sum^{n_e}_{i=2}(\varepsilon_i+\varepsilon'_i)\\
     &=2|\mathcal{G}_m|+n_e-2-\sum^{n_e}_{i=1}(\varepsilon_i+\varepsilon'_i)\\
     &=2|\mathcal{G}_m|+n_e-2-b_2
\end{aligned}.
\end{equation}

Subtracting (\ref{eq_linkNodeRelation}) from (\ref{eq_totalCycles}) and utilizing (\ref{eq_red_num}), we get
\begin{equation}\label{eq_concCyclNum}
    |C|=||\mathcal{G}_m||.
\end{equation}

Let $\mathcal{G}'_{m}:=\mathcal{G}_m$ except all virtual links/nodes in $\mathcal{G}_m$ are real links/nodes in $\mathcal{G}'_{m}$. Suppose cycle measurement is allowed and $r$ is the only monitor in $\mathcal{G}'_{m}$, then all cycle measurements associated with $C$ is sufficient to identify all links in $\mathcal{G}'_{m}$ (following the same method in STLI, see \cite{LiangPathConstruction12}). Moreover, we have $||\mathcal{G}_m||=||\mathcal{G}'_{m}||$; therefore, all cycles in $C$ are linearly independent.
\end{proof}

\subsection{Number of Linearly Independent Paths After Removing All Virtual Links}
\label{sec_NumLinearIndePaths}
In section \ref{sec_numDisCycl}, we get a cycle set $C$ with $||\mathcal{G}_m||$ linearly independent cycles. For any cycle in $C$, removing all involved virtual links and the resulting isolated nodes, we can obtain a monitor-to-monitor simple path since each cycle in $C$ is obtained by combining two internally vertex disjoint paths ($\mathcal{X}^{(i)}$ and $\mathcal{X}^{(j)}$) ($i,j\in\{1,2,3\}$, $i\neq j$) along two independent spanning trees and virtual links only appear at the end of $\mathcal{X}^{(i)}$ and $\mathcal{X}^{(j)}$. Therefore, following the same operation to each cycle in $C$ iteratively, we can generate a new set ${Y}$ with each element $\mathcal{Y}_i$ representing a monitor-to-monitor simple paths. In this section, we investigate on the number of linearly independent paths in ${Y}$. Note that some paths in ${Y}$ might contain more than 2 monitors. We will show how to shorten these paths in the next section.

Let $\mathbb{V}$ denote the set of virtual links\footnote{$\mathcal{G}_m$ is a spanning graph of $\mathcal{G}^*_{ex}$. Thus, the number of virtual links in $\mathcal{G}_m$ is less than that in $\mathcal{G}^*_{ex}$.} in $\mathcal{G}_m$. The goal is to prove the number of linearly independent paths in ${Y}$ is $||\mathcal{G}_m||-|\mathbb{V}|$, and these linearly independent paths are sufficient to identify all tree links in the original graph $\mathcal{G}$.

For all redundant paths generated by mapping from $C$ to ${Y}$, they can be divided into 3 categories, i.e., 1) trivial topology, 2) linearly dependent paths, and 3) duplicate paths. Note in Fig.~\ref{Fig:TrivialTopo}--Fig.~\ref{Fig:OtherVirtualLinks}, for all paths (or path segments) colored as blue/green/red, they represent virtual links if they are outside $\mathcal{G}$; and real simple paths otherwise.
\begin{figure}[tb]
\centering
\includegraphics[width=1.5in]{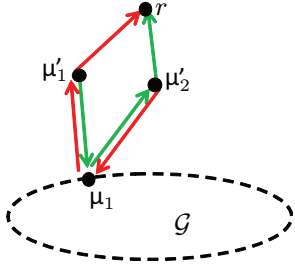}
\caption{Combining red and green paths wrt ${{\mu}_1}$.} \label{Fig:TrivialTopo}
\end{figure}

1) Trivial topology (a graph with no nodes or links). As shown in Fig.~\ref{Fig:TrivialTopo}, combining the red and green paths wrt ${\mu}_1$ ($\mathcal{G}_m$ involves $r{\mu'}_1{\mu}_1{\mu}'_2r$), cycle $r{\mu'}_1{\mu}_1{\mu}'_2r$ is formed. However, the corresponding path of $r{\mu'}_1{\mu}_1{\mu}'_2r$ after removing the virtual links and the resulting isolated nodes is empty. Therefore, this trivial path has no contribution for identifying real links in $\mathcal{G}$.

2) Linearly dependent paths. The following 3 cases can generate redundant linearly dependent paths.
\begin{figure}[tb]
\centering
\includegraphics[width=1.5in]{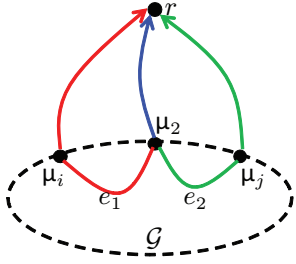}
\caption{Path construction wrt ${{\mu}_2}$.} \label{Fig:LinearDepPaths1}
\end{figure}

(i) Consider the three cycles constructed wrt ${\mu}_2$ (shown in Fig.~\ref{Fig:LinearDepPaths1}). The paths associated with $blue+red$ cycle and $blue+green$ cycle are ${\mu}_ie_1{\mu}_2$ and ${\mu}_2e_2{\mu}_j$, respectively. However, the path associated $red+green$ cycle, ${\mu}_ie_1{\mu}_2e_2{\mu}_j$, is the sum of previously constructed paths ${\mu}_ie_1{\mu}_2$ and ${\mu}_2e_2{\mu}_j$. Therefore, redundant path ${\mu}_ie_1{\mu}_2e_2{\mu}_j$ is linearly dependent with the other two.
\begin{figure}[tb]
\centering
\includegraphics[width=1.8in]{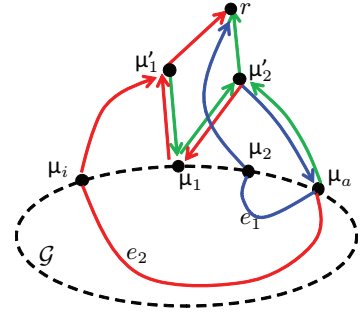}
\caption{Path construction wrt ${{\mu}_a}$.} \label{Fig:LinearDepPaths2}
\end{figure}

(ii) Now consider the case shown in Fig.~\ref{Fig:LinearDepPaths2}. For $\mathcal{G}_m$, there exists one and only one monitor ${\mu}_a$ with link ${\mu}'_2{\mu}_a$ colored as blue in the ${\mu}'_2\rightarrow {\mu}_a$ direction, according to the rules for blue tree construction. Now consider ${\mu}_a$. According to the s-t numbering rule and the processing of ear decomposition, we have $f({\mu}'_2)<f({\mu}_a)$ and $g({\mu}'_2)<g({\mu}_a)$; therefore, ${\mu}'_2{\mu}_a$ is colored\footnote{Note ${\mu}_a$ might have more than one neighbor $w$ with $f(w)<f({\mu}_a)$ and $g(w)<g({\mu}_a)$, in which case we can still choose ${\mu}'_2{\mu}_a$ to color it as green.} as green in the other direction. Then ${\mu}_2e_1{\mu}_a$ and ${\mu}_ie_2{\mu}_a$ are two paths obtained from $green+blue$ cycle and $green+red$ cycle wrt ${\mu}_a$, respectively. However, the path generated by the $blue+red$ cycle is the sum of ${\mu}_2e_1{\mu}_a$ and ${\mu}_ie_2{\mu}_a$, thus redundant in ${Y}$.
\begin{figure}[t]
\centering
\includegraphics[width=1.8in]{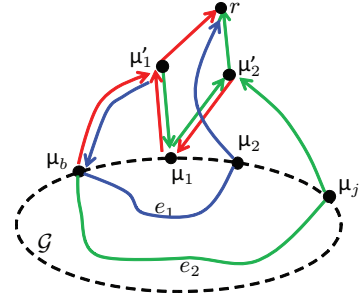}
\caption{Path construction wrt ${{\mu}_b}$.} \label{Fig:LinearDepPaths3}
\end{figure}

(iii) Analogously, there exists one and only one monitor ${\mu}_b$ with link ${\mu}'_1{\mu}_b$ colored as blue in the ${\mu}'_1\rightarrow {\mu}_b$ direction (see Fig.~\ref{Fig:LinearDepPaths3}). Based on the nonseparating ear decomposition, we have $f({\mu}'_1)>f({\mu}_b)$ and $g({\mu}'_1)<g({\mu}_b)$; therefore, ${\mu}'_1{\mu}_b$ is colored\footnote{Note ${\mu}_b$ might have more than one neighbor $u$ with $f(u)>f({\mu}_a)$ and $g(u)<g({\mu}_a)$, in which case we can still choose ${\mu}'_1{\mu}_b$ to color it as red.} as red in the other direction. Then ${\mu}_be_1{\mu}_2$ and ${\mu}_be_2{\mu}_j$ are two paths obtained from $red+blue$ cycle and $red+green$ cycle wrt ${\mu}_b$, respectively. However, the path generated by the $blue+green$ cycle is the sum of ${\mu}_be_1{\mu}_2$ and ${\mu}_be_2{\mu}_j$, thus not providing new information for link identifications in $\mathcal{G}$.

3) Duplicate paths. Each of the following 3 cases can generate the same path twice.
\begin{figure}[tb]
\centering
\includegraphics[width=1.8in]{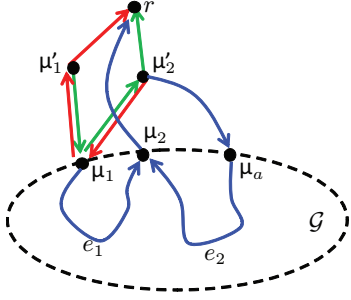}
\caption{Duplicate paths wrt ${{\mu}_1}$ and ${\mu}'_2$.} \label{Fig:DupPath12}
\end{figure}

(i) We have considered combining the red and green paths wrt ${\mu}_1$. Now consider the paths associated with the $blue+red$ cycle and $blue+green$ cycle. As Fig.~\ref{Fig:DupPath12} displays, the paths associated with these two cycles are exactly the same, i.e., path ${\mu}_1e_1{\mu}_2$ is generated twice.

(ii) Following the similar argument, the paths associated with the $blue+red$ cycle and $blue+green$ cycle wrt ${\mu}'_2$ are also the same, i.e., path ${\mu}_2e_2{\mu}_a$ (shown in Fig.~\ref{Fig:DupPath12}) is generated twice when mapping from $C$ to ${Y}$. Note ${\mu}_a$ cannot be the same as ${\mu}_2$ since ${\mu}_2$ only appears in the last ear. Therefore, path ${\mu}_2e_2{\mu}_a$ is non-trivial.
\begin{figure}[tb]
\centering
\includegraphics[width=1.8in]{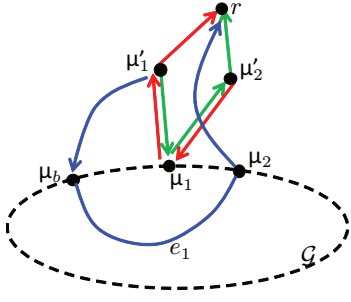}
\caption{Duplicate paths wrt ${\mu}'_1$.} \label{Fig:DupPath3}
\end{figure}

(iii) Similar to (i) and (ii), the paths associated with the $blue+red$ cycle and $blue+green$ cycle wrt ${\mu}'_1$ are also the same, i.e., ${\mu}_be_1{\mu}_2$ as shown in Fig.~\ref{Fig:DupPath3} appears twice when mapping from $C$ to ${Y}$. In addition, ${\mu}_be_1{\mu}_2$ is non-trivial since ${\mu}_2$ has the highest ear level which means ${\mu}_b\neq {\mu}_2$.
\begin{figure}[tb]
\centering
\includegraphics[width=3.3in]{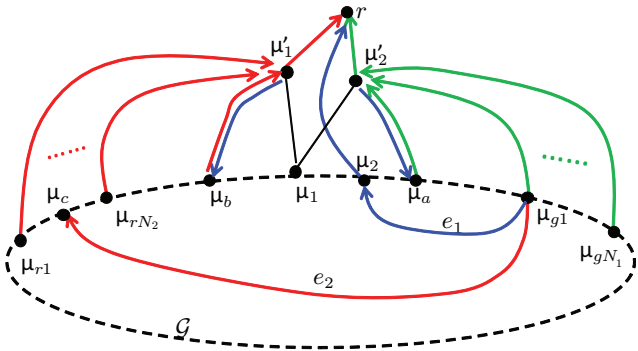}
\caption{Other virtual links connecting to ${\mu}'_1$ or ${\mu}'_2$.} \label{Fig:OtherVirtualLinks}
\end{figure}

4) We have discussed 7 cases in 1)--3), each case providing a redundant path when mapping from $C$ to ${Y}$. For the spanning graph $\mathcal{G}_m$, the minimum number of virtual links in $\mathcal{G}_m$ is also 7, i.e., $r{\mu'}_1$, $r{\mu'}_2$, ${\mu}'_1{\mu}_1$, ${\mu}'_2{\mu}_1$, ${\mu}'_2{\mu}_a$ (Fig.~\ref{Fig:LinearDepPaths2}), and ${\mu}'_1{\mu}_b$ (Fig.~\ref{Fig:LinearDepPaths3}). However, in addition to these 7 virtual links, there might exist other virtual links in $\mathcal{G}_m$. As Fig.~\ref{Fig:OtherVirtualLinks} shows, suppose there are other $N_1$ virtual links (colored as green and no color on the other direction as ${\mu}'_1{\mu}_b$ has been colored as blue) connecting ${\mu}'_2$ and ${\mu}_{gi}$ ($i=1,\cdots,N_1$) and $N_2$ virtual links (colored as red and no color on the other direction as ${\mu}'_1{\mu}_b$ has been colored as blue) connecting ${\mu}'_1$ and ${\mu}_{ri}$ ($i=1,\cdots,N_2$). For the 3 paths constructed wrt each node in $\{{\mu}_{g1},\cdots,{\mu}_{gN_1},{\mu}_{r1},\cdots,{\mu}_{rN_2}\}$, there exists one path which is linearly dependent with the other two. To prove this claim, consider ${\mu}_{g1}$ as an example. Two paths\footnote{Note ${\mu}_c$ can be any node in $\{{\mu}_{r1},\cdots,{\mu}_{rN_2},{\mu}_b\}$ in Fig.~\ref{Fig:OtherVirtualLinks}.} ${\mu}_2e_1{\mu}_{g1}$ and ${\mu}_ce_2{\mu}_{g1}$ (see Fig.~\ref{Fig:OtherVirtualLinks}) can be obtained from the $green+blue$ cycle and $green+red$ cycle wrt ${\mu}_{g1}$. However, the path associated with the $blue+red$ cycle wrt ${\mu}_{g1}$ is the sum of ${\mu}_2e_1{\mu}_{g1}$ and ${\mu}_ce_2{\mu}_{g1}$, thus linearly dependent with the other two. The same argument can be applied to other nodes in $\{{\mu}_{g1},\cdots,{\mu}_{gN_1},{\mu}_{r1},\cdots,{\mu}_{rN_2}\}$. Therefore, we can identify another $N_1+N_2$ linearly dependent paths in ${Y}$ regarding these $N_1+N_2$ virtual links.

In sum, 1)--4) cover all the possible cases of redundant paths when mapping from $C$ to ${Y}$ and the number of these redundant paths is $7+N_1+N_2$, which is also the number of virtual links (denoted by $|\mathbb{V}|$) in $\mathcal{G}_m$. Therefore, removing these redundant paths, the cardinality of the resulting path set ${Y}'$ is $||\mathcal{G}_m||-(7+N_1+N_2)=||\mathcal{G}_m||-|\mathbb{V}|$.

Now we can explore if ${Y}'$ is sufficient to identify all tree links in $\mathcal{G}^*_{ex}$. Recall that each path (for identifying tree links) constructed by STPC consists of one or two path segments, each path segment terminating at the first monitor it encounters. If we let each segment terminate at the last monitor before traversing a virtual link and combine any two segments wrt real node $v$ ($v$ can be either monitor or non-monitor), then a new measurement path set $J$ is formed. It is easy to show (using STLI) that this new path set $J$ is sufficient to identify all tree links in $\mathcal{G}^*_{ex}$. Recall the mapping process from cycle set $C$ to path set ${Y}$. Each path in ${Y}$ is obtained by removing the virtual links and all resulting isolated nodes of the corresponding cycle in $C$. Then each path $\mathcal{P}_J$ in $J$ can be obtained by applying the same operation to the corresponding cycle $C_J$ ($C_J\supset \mathcal{P}_J$) in $C$; therefore, each path in $J$ must be involved as one element in ${Y}$. Accordingly, ${Y}$ is sufficient to identify all tree links in $\mathcal{G}^*_{ex}$. Since ${Y}'$ is obtained by removing redundant paths in the procedure of mapping $C$ to ${Y}$, then ${Y}'$ is also sufficient to identify all tree links in $\mathcal{G}^*_{ex}$, i.e., the measurement matrices associated with $J$, ${Y}$ and ${Y}'$ have the same column rank.

In sum, for graph $\mathcal{G}_m$ with $||\mathcal{G}_m||-|\mathbb{V}|$ real links, the property of ${Y}'$ is that it contains exactly $||\mathcal{G}_m||-|\mathbb{V}|$ paths, with each path containing only tree links in $\mathcal{G}^*_{ex}$. Therefore, the measurement matrix associated with ${Y}'$ is a square matrix with full rank.

\subsection{Number of Distinct Paths by STPC}
\begin{theorem}\label{thm:number of distinct path}
The number of distinct paths constructed by STPC equals $n$, the number of links in $\mathcal{G}$.
\end{theorem}

\begin{proof}
We have proved that the corresponding measurement matrix of ${Y}'$ is square and sufficient to identify all tree links in $\mathcal{G}^*_{ex}$. In this section, we prove that all paths with more than 2 monitors in ${Y}'$ can be shortened to form a new path set ${Z}$ which can cover all the paths obtained by STPC for tree link identification, i.e., paths obtained by line 1--11 in STPC algorithm.

1) First, we consider the number of generated paths wrt to node $v$ (line 5--9 in STPC).

(i) If $v$ is a monitor in $\{{\mu}_2,{\mu}_a,{\mu}_b,{\mu}_{g1},\cdots,{\mu}_{gN_1},{\mu}_{r1},\cdots,{\mu}_{rN_2}\}$ (see Fig.~\ref{Fig:OtherVirtualLinks}), then one path in $\{\mathcal{P}_{v1},\mathcal{P}_{v2},\mathcal{P}_{v3}\}$ contains no links, resulting to be an invalid path and discarded without appending to $\mathbf{R}$ in line 10. In this case, only two paths are generated in line 6. In 2) and 4) of Section \ref{sec_NumLinearIndePaths}, we have shown that one linearly dependent path, which is the combination of the other two paths wrt the same node, has been removed, thus not existed in ${Y}'$. Therefore, in ${Y}'$, there are also two paths wrt to nodes belonging to $\{{\mu}_2,{\mu}_a,{\mu}_b,{\mu}_{g1},\cdots,{\mu}_{gN_1},{\mu}_{r1},\cdots,{\mu}_{rN_2}\}$.

(ii) If $v$ is ${\mu}_1$, then line 6 of STPC only generates one path, since the other two paths contain only one node, i.e., ${\mu}_1$ itself. In this case, we have also shown that only path (${\mu}_1e_1{\mu}_2$ colored as blue in Fig.~\ref{Fig:DupPath12}) wrt ${\mu}_1$ is retained in ${Y}'$ after removing one invalid (Section \ref{sec_NumLinearIndePaths}-1)) and one duplicate (Section \ref{sec_NumLinearIndePaths}-3)-(i)) path.

(iii) Apart from the above two cases, each node in $\mathcal{G}$ corresponds to 3 measurement paths both in ${Y}'$ and $\mathcal{P}_{STPC}$, where $\mathcal{P}_{STPC}$ is the path set by STPC after executing line 1--11.
\begin{figure}[tb]
\centering
\includegraphics[width=2.3in]{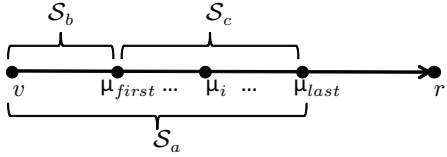}
\caption{Shortening a long path in ${Y}'$.} \label{Fig:LongShortPaths}
\end{figure}

2) Second, we compare the differences between the two path segments obtained by the two methods wrt to node $v$ on the same tree $\mathcal{T}_i$. Let $\mathcal{S}_a$ and $\mathcal{S}_b$ denote the corresponding path segments associated with ${Y}'$ and $\mathcal{P}_{STPC}$, respectively. Both $\mathcal{S}_a$ and $\mathcal{S}_b$ start at $v$ and go toward $r$ along the same tree $\mathcal{T}_i$. The only difference is that $\mathcal{S}_a$ terminates at the last monitor before traversing a virtual link and $\mathcal{S}_b$ terminates at the first monitor it encounters (as shown in Fig.~\ref{Fig:LongShortPaths}). In Fig.~\ref{Fig:LongShortPaths}, observe that wrt ${\mu}_{first}$, $\mathcal{S}_c$ is constructed and $W_{\mathcal{S}_c}$ can be computed through measuring the paths constructed wrt ${\mu}_{first}$ in ${Y}'$. Therefore, when constructing path segment $\mathcal{S}_a$ wrt $v$, there is no need to terminate at ${\mu}_{last}$. $\mathcal{S}_a$ can also terminate at ${\mu}_{first}$ since $W_{\mathcal{S}_c}$ can be known from other path measurements directly (there are only two monitors on $\mathcal{S}_c$) or indirectly (in Fig.~\ref{Fig:LongShortPaths}, $W_{{{\mu}_i\rightarrow {\mu}_{last}}}$ and $W_{{{\mu}_{first}\rightarrow {\mu}_{i}}}$ are obtained from paths constructed wrt ${\mu}_i$ and ${\mu}_{first}$, respectively). Using this operation, each path with more than 2 monitors in ${Y}'$ can be shortened.

3) Let ${Y}''$ denote the path set obtained from ${Y}'$ after the above shortening process. We can show in set ${Y}''$, for a monitor ${\mu}^*$ which is not in $\{{\mu}_2,{\mu}_a,{\mu}_b,{\mu}_{g1},\cdots,{\mu}_{gN_1},{\mu}_{r1},\cdots,{\mu}_{rN_2}\}$, the corresponding measurement paths can be further shortened. Let $\mathcal{S}^*_1$, $\mathcal{S}^*_2$ and $\mathcal{S}^*_3$ denote the corresponding shortened path segments wrt ${\mu}^*$, then the associated measurement paths in ${Y}''$ are $\mathcal{S}^*_1\cup \mathcal{S}^*_2$, $\mathcal{S}^*_2\cup \mathcal{S}^*_3$ and $\mathcal{S}^*_3\cup \mathcal{S}^*_1$, which can be further shortened to $\mathcal{S}^*_1$, $\mathcal{S}^*_2$ and $\mathcal{S}^*_3$, respectively, forming the same 3 paths as those obtained by STPC. Suppose the newly formed set from ${Y}''$ by path shortening is ${Z}$, then ${Z}$ is also sufficient to identify all tree links in $\mathcal{G}_m$. This is because, from the perspective of linear algebra, this shortening operation means that one part of the original linear equation can be eliminated by subtracting the linear combinations of some other rows; therefore, the resulting matrix rank is the same as the original matrix. Since we have proved that the measurement matrix associated with ${Y}'$ has a full rank, the measurement matrix associated with ${Z}$ also has a full rank.

To identify tree links in $\mathcal{G}^*_{ex}$, based on the above three arguments, we know that the number of paths constructed wrt $v$ are the same in both ${Y}'$ and $\mathcal{P}_{STPC}$. Moreover, wrt the same node and along the same independent spanning tree(s), the paths obtained by these two methods are exactly the same after the path shortening operation. Thus, the newly generated path set ${Z}$ can cover all the paths selected by STPC for tree link identification. The number of paths in ${Z}$ is the number of tree links in $\mathcal{G}^*_{ex}$; therefore, the number of distinct paths constructed by line 1--11 in STPC is exactly the number of tree links in $\mathcal{G}^*_{ex}$.

Finally, we consider the non-tree links of the original graph $\mathcal{G}$. As $\mathcal{G}_m=\mathcal{T}_1 \cup \mathcal{T}_2 \cup \mathcal{T}_3$, all non-tree links are involved in set $L(\mathcal{G}\setminus \mathcal{G}_m)$. With the knowledge of tree link metrics, based on the auxiliary algorithm of STPC for non-tree link identification, in line 12--15, each non-tree link $l$ (with link metric $W_l$) corresponds to one new path measurement involving $W_l$ as the only unknown link metric. Therefore, the number of newly added paths by line 12--15 equals the number of non-tree links in $\mathcal{G}$ and they are linearly independent with all other selected paths. A link in $\mathcal{G}$ is either a tree or non-tree link; therefore, the number of distinct paths constructed by STPC equals the number of links in $\mathcal{G}$.
\end{proof}

\section{Performance Evaluations on Random Graphs}
Besides the main simulation results in \cite{LiangPathConstruction12}, we also simulate random graphs with a different number of links and observe similar results in Table \ref{t sparse}.

%\begin{comment}
\begin{table*}[tb]
%\vspace{-.5em}
\renewcommand{\arraystretch}{1.3}
\caption{Sparsely-Connected Random Graphs (ER: $p=0.0390$, RG: $d_c=0.11943$, BA: $\varrho=3$, $I_{\mbox{\tiny max}}=3\times n$)} \label{t sparse}
\vspace{-.5em}
\centering
\begin{tabular}{c|c|c|c|c|c|c|c|c|c|c|c}
  \hline
   graph & $\overline{n}$  & ${m}$ & $\overline{\kappa}$ & $r_{\mbox{\small succ}}$ & ${\Upsilon}$ & $\tSTPC$ (s)& $\tRWPC$ (s)& $\tSTLI$ (ms)& $\tMILI$ (ms)& $\hSTPC$& $\hRWPC$ \\
  \hline
ER	&	438.48	& 150 &	9.39	&	99.00\%	&	99.76\%	&	10.21	&	38.73	&	 5.27	&	17.8	&	17.94	&	 14.2	 \\
\hline																					
RG	&	449.02	& 150 &	14.96	&	4.00\%	&	93.91\%	&	10.62	&	109.21	&	 5.10	&	23.34	&	22.48	&	 14.43	 \\
\hline																					
BA	&	441	& 150 &	3	&	72.00\%	&	99.70\%	&	8.41	&	90.48	&	5.39	 &	 20.45	&	15.08	&	8.85	 \\
\hline																																			 
\end{tabular}
\vspace{-0mm}
\end{table*}
%\end{comment}

\bibliographystyle{IEEEtran}
\bibliography{mybibSimplified}

\end{document}